\documentclass[a4paper,12pt]{article}
\usepackage{setspace}
\usepackage[nomarkers]{endfloat} 

\usepackage{titlesec}
\titleformat*{\section}{\center \bf}
\titleformat*{\subsection}{\raggedright \bf}

\usepackage[margin=0.8in]{geometry}
\usepackage{indentfirst}
\usepackage{graphicx}
\usepackage{ dsfont }
\usepackage{amsmath}
\usepackage{amssymb}
\usepackage{amsopn}
\usepackage{stfloats}
\usepackage{hyperref}
\usepackage{cite}

\hypersetup{colorlinks=true,citecolor=blue}

\usepackage{xcolor}

\usepackage[english]{babel}
\newtheorem{theorem}{Theorem}
\newtheorem{lemma}[theorem]{Lemma}
\newenvironment{proof}[1][Proof]{\begin{trivlist}
\item[\hskip \labelsep {\bfseries #1}]}{\end{trivlist}}
\newcommand{\qed}{\nobreak \ifvmode \relax \else
      \ifdim\lastskip<1.5em \hskip-\lastskip
      \hskip1.5em plus0em minus0.5em \fi \nobreak
      \vrule height0.75em width0.5em depth0.25em\fi}

\title{Simplified Compute-and-Forward and Its Performance Analysis}
\author{
Mohsen Hejazi\\
\small Department of Electrical Engineering\\[-0.8ex]
\small Sharif University of Technology\\
\small Tehran, Iran\\
\small \texttt{mhejazi@ee.sharif.edu}\\
\and
Masoumeh Nasiri-Kenari\\
\small Department of Electrical Engineering\\[-0.8ex]
\small Sharif University of Technology\\
\small Tehran, Iran\\
\small \texttt{mnasiri@sharif.edu}\\
}

\date{}

\begin{document}
\maketitle

\begin{abstract}
\textbf{
The compute-and-forward (CMF) method has shown a great promise as an innovative approach to exploit interference toward achieving higher network throughput. The CMF was primarily introduced by means of information theory tools. While there have been some recent works discussing different aspects of efficient and practical implementation of CMF, there are still some issues that are not covered. In this paper, we first introduce a method to decrease the implementation complexity of the CMF method. We then evaluate the exact outage probability of our proposed simplified CMF scheme, and hereby provide an upper bound on the outage probability of the optimum CMF in all SNR values, and a close approximation of its outage probability in low SNR regimes.
We also evaluate the effect of the channel estimation error (CEE) on the performance of both optimum and our proposed simplified CMF by simulations. Our simulation results indicate that the proposed method is more robust against CEE than the optimum CMF method for the examples considered.}
\par
\textbf{
Index Term- compute-and-forward, outage probability, performance analysis, channel estimation error, wireless relay network.}
\end{abstract}


\section{INTRODUCTION} \label{sec:Intro}
The CMF method, proposed by Nazer and Gastpar~\cite{N1}, enables exploiting, rather than combating, the multiple access interference in a wireless relay network, and thus results in improved network throughput~\cite{N2}. In this method, relays, instead of recovering single messages, attempt to reliably recover and pass an integer linear combination of transmitted messages, called an equation, to the destination. By receiving enough equations, the destination can solve the linear equation system to recover desired messages.\par
In recent years, the CMF method as a promising approach, has received a lot of attention. There have been an increasing number of works on theoretical aspects of CMF method. In~\cite{E1}, by analyzing asymptotic behavior, the number of degrees of freedom of CMF method is derived. The authors in~\cite{E3} generalize the CMF method to the case of multiple antenna sources and relays, and compute the corresponding achievable rate based on an optimization problem. Successive recovering of messages in relays for CMF method is proposed in~\cite{E2}. The idea is similar to successive interference cancellation (SIC) technique and leads to a higher recovering rate at each relay. In~\cite{ISI}, the impact of inter-symbol interference on the CMF method, for the bi-directional relay case, is evaluated. \par
Most of the analytical studies, including aforementioned works, consider the CMF from an information theory viewpoint which involves finding an achievable rate or a capacity region~\cite{MAC, hop, SUM}. On the other hand, some works concentrate on the practical aspects and implementation of the CMF~\cite{Flatness}. Efficient design of proper codes and lattices are discussed in~\cite{E5, E6, E7, E9}. Reference~\cite{E10} considers the problem of lattice decoding and proposes some practical and efficient approaches to this end. Finding an optimum integer Equation Coefficient Vector (ECV), coefficients of recovered combination at a relay, is computationally complex. This issue in not addressed considerably and will be investigated in this paper.\par
The authors in~\cite{N1} and~\cite{MWRC} have calculated the outage performance of the CMF method in the cases of three-transmitter multiple-access channel (MAC) and Multi-Way relay channels, respectively. However, in these works, the achievable rate of CMF has not been given by a closed-form expression, but rather as the solution of an optimization problem ~\cite{N1}. In fact, to the best of our knowledge, all the previous works, including~\cite{N1} and~\cite{MWRC}, have employed numerical calculations or simulations as a part of their analysis to find the outage probability.\par
In this paper, we first present some lemmas which help us to reduce the complexity of optimal CMF method. These lemmas are then used to approximate the optimum CMF as a simplified CMF with considerably less computational complexity. We provide the exact outage probability of our proposed scheme and show that our proposed method performs near the optimum CMF in low SNR values. We also consider the effect of the channel estimation error (CEE) on the performance of both the optimum and simplified CMF methods. our simulation results show that the proposed scheme is more robust against the CEE than the optimum scheme for the examples considered.\par
The rest of this paper is organized as follows. Section~\ref{sec:system} introduces the system model and assumptions. The optimum and the proposed CMF methods are presented in section~\ref{sec:ECV}. Section~\ref{sec:analysis} includes the performance analysis of the proposed method. Simulation and numerical results are presented in section~\ref{sec:simulation}. Finally, the paper is concluded in section~\ref{sec:conclusion}.\par

\section{SYSTEM MODEL} \label{sec:system}
We consider a general network, which itself can be a part of a larger network, consisting of $L=2$ transmitters as message sources, $M$ relays, and one receiver (destination), as shown in Fig.~\ref{fig:Model}. The destination aims to reliably recover both messages of the sources.\par
\begin{figure}[t]
\renewcommand{\figurename}{Fig.}
\centering
\includegraphics[width = \columnwidth]{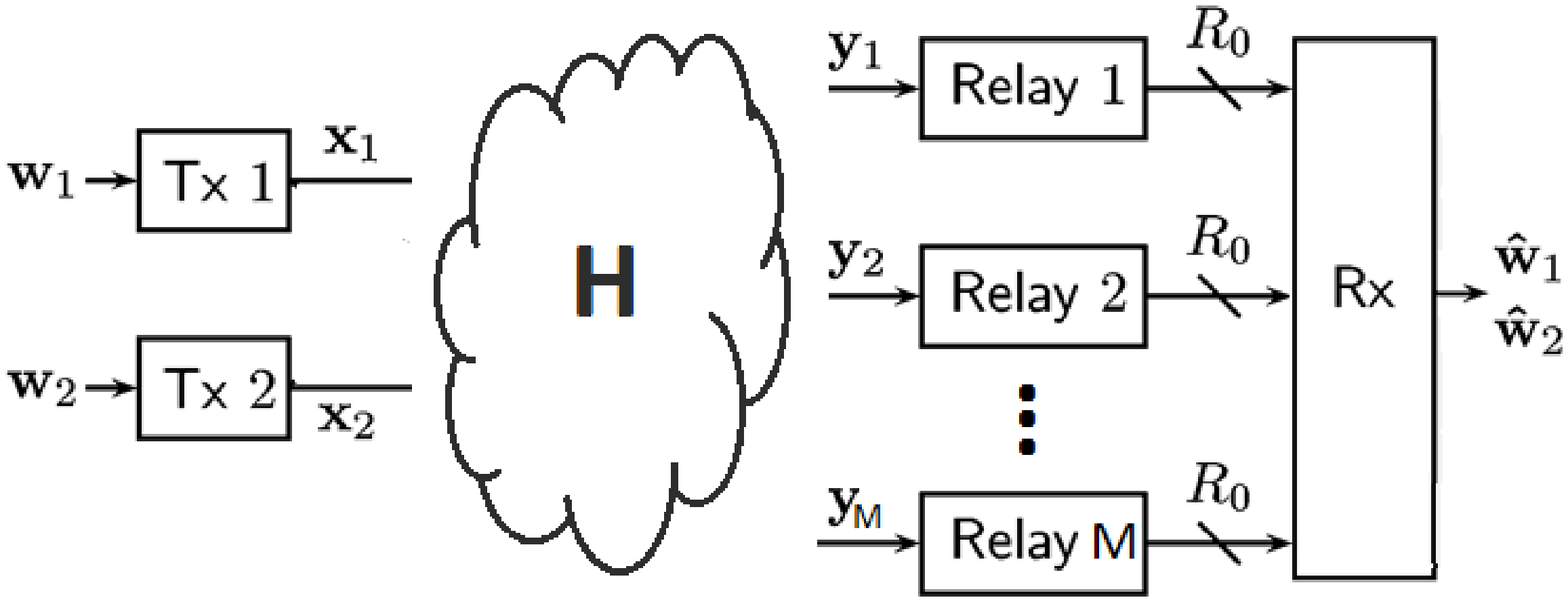}
\caption{System model.}
\label{fig:Model}
\end{figure}
Let the channel gain from source $l$ to relay $m$ and the channel gain vector of relay $m$ be denoted by $h_{ml}$ and  ${{\bf{h}}_m} = {[{h_{m1}},{h_{m2}}]^T}$, respectively. We consider block fading channels, and assume that the channel gains are real independent and identically Rayleigh distributed (\textit{i.i.d.}) variables with unit variance. The channel noises are additive white Gaussian (AWGN) with unit variance. Let $P_l$ be the transmission power of source $l$. Then, the average SNR received at each relay from source $l$ is equal to $P_l$, when not considering the shadowing effects or assuming the same shadowing effects for all relays. The results can be easily extended to more general cases. The channels from relays to the destination are assumed to be orthogonal and have high enough capacity ($R_0$ in Fig.~\ref{fig:Model}) to reliably transfer the required information.\par
Our proposed method of transmission over the network is based on the CMF scheme. A detailed description of CMF scheme can be found in~\cite{N1}. First, two sources map their messages ${\mathbf{w}}_1$ and ${\bf{w}}_2$ to symbols ${\bf{x}}_1$ and ${\bf{x}}_2$, respectively, and transmit the symbols simultaneously. Receiving a noisy linear combination of transmitted symbols ${\bf{y}}_m$, the $m$-th relay, $m=1,2,…,M$, computes an equation (integer linear combination) of transmitted messages with ECV equal to  ${{\mathbf{a}}_m} = {[{a_1},{a_2}]^T} \in {\mathbb{Z}^2}$.Then, all relays pass their computed equations and ECVs to the destination. Finally, the destination attempts to extract both messages from $M$ ($M\geq2$) received equations. Specifically, the destination selects two independent equations with the highest rates and solves them to recover both the messages.\par

\section{EFFICIENT ECV} \label{sec:ECV}
We attempt to find an efficient ECV in the relay that provides the highest computation rate (the rate of recovering an equation) while imposes tolerable complexity on the system. A relay $m$ with channel vector ${\mathbf{h}}_m$ can recover an equation with ECV equal to $\mathbf{a}$, as long as the message rates are less than the computation rate~\cite{N1} defined as
\begin{equation}
R({{\mathbf{h}}_m},{\mathbf{a}}) = \,\frac{1}{2}{\log ^ + }\left( {{{\left( {{{\left\| {\mathbf{a}} \right\|}^2} - \frac{{P{{\left| {{{\mathbf{a}}^T}{{\mathbf{h}}_m}} \right|}^2}}}{{1 + P{{\left\| {{{\mathbf{h}}_m}} \right\|}^2}}}} \right)}^{ - 1}}} \right),
\end{equation}
where the two sources have the same power $P$, i.e. $P_1=P_2=P$, and
\begin{equation}
\log ^ + (x) \triangleq \max (x,0).
\end{equation}\par
 If we define the vector ${\mathbf{g}}_m$ as
\begin{equation}\label{eq:gmVec}
{{\mathbf{g}}_m} \triangleq {\left[ {{g_{m1}},{g_{m2}}} \right]^T} = {\left[ {{h_{m1}}\sqrt {{P_1}} ,{h_{m2}}\sqrt {{P_2}} } \right]^T},
\end{equation}\
computation rate for unequal source powers can be written as
\begin{equation}\label{eq:Rga}
R({{\mathbf{g}}_m},{\mathbf{a}}) = \,\frac{1}{2}{\log ^ + }\left( {{{\left( {{{\left\| {\mathbf{a}} \right\|}^2} - \frac{{{{\left| {{{\mathbf{a}}^T}{{\mathbf{g}}_m}} \right|}^2}}}{{1 + {{\left\| {{{\mathbf{g}}_m}} \right\|}^2}}}} \right)}^{ - 1}}} \right).
\end{equation}
It is notable that $g_{ml}^2$ is the instantaneous received SNR from source $l$ at the relay and hence, ${\left\| {{{\mathbf{g}}_m}} \right\|^2}$ (called instantaneous sum SNR) equals the sum of all instantaneous received SNRs. 
In the following, the optimum method along with our proposed suboptimum method of selecting ECVs is presented.\par

\subsection{Optimum Compute-and-Forward}
In the optimum CMF method~\cite{N1}, each relay selects the ECV with the highest computation rate; More specifically relay $m$ calculates its ECV according to the following maximization problem
\begin{equation}
 {{\mathbf{a}}_m} = \arg \max_{\substack{{\mathbf{a}} \in {\mathbb{Z}^2}\\{\mathbf{a}} \ne {\mathbf{0}}}} R({{\mathbf{g}}_m},{\mathbf{a}}),
\end{equation}
where is equivalent to~\cite{Flatness}
\begin{equation} \label{eq:opt2}
 {{\mathbf{a}}_m} = \arg \min_{\substack{{\mathbf{a}} \in {\mathbb{Z}^2}\\{\mathbf{a}} \ne {\mathbf{0}}}} {{\mathbf{a}}^T}{{\mathbf{G}}_m}{\mathbf{a}},
\end{equation}
where ${\mathbf{g}}_m$ is defined in (\ref{eq:gmVec}) and ${\mathbf{G}}_m$ is a positive-definite matrix and is defined as
\begin{equation}\label{eq:gmMat}
{{\mathbf{G}}_m} \triangleq \left( {{\mathbf{I}} - \frac{{{{\mathbf{g}}_m}{{\mathbf{g}}_m}^T}}{{1 + {{\left\| {{{\mathbf{g}}_m}} \right\|}^2}}}} \right).
\end{equation}\par
The above integer optimization problem is equivalent to the shortest vector problem (SVP), and has no closed-form solution~\cite{Flatness, MWRC}. Lattice reduction algorithms such as~\cite{Fin} can be applied to calculate the optimum ECV numerically. However, due to the large search space of the problem on hand, the complexity of the system is considerably high. In the following, we introduce some techniques to reduce the size of search space and then propose our simplified compute-and-forward method. \par

\subsection{Reducing the search space}
The optimization in~(\ref{eq:opt2}) is defined over $\mathbb{Z}^2$ and has a infinite search space. To limit the size of search set, which leads to reduced complexity, we use the following lemmas. Note that all the following lemmas are also hold for the more general case when the number of sources are larger than 2 ($L \geq 2$). In fact for this general case, we have ${{\bf{g}}_m} = {[{g_{m1}},\cdots,{g_{mL}}]^T}$ and ${{\mathbf{a}}_m} = {[{a_1},\cdots,{a_L}]^T} \in {\mathbb{Z}^L}$. Moreover, although we consider real vectors, the complex case can be modeled as a real one with $2L$ sources and $2M$ relays~\cite{N1}. \par

\begin{lemma} \label{lem:1}
 (Lemma 1 in~\cite{N1}) An ECV ${\mathbf{a}}_m$ results in a zero computation rate for 
\begin{equation}
	{\left\| {{{\mathbf{a}}_m}} \right\|^2} \geqslant 1 + {\left\| {{{\mathbf{g}}_m}} \right\|^2}.
\end{equation}
\end{lemma}\par
From this lemma, for the optimization problem of~(\ref{eq:opt2}), as stated in~\cite{N1}, it is sufficient to search only over ECVs that satisfy ${\left\| {{{\mathbf{a}}_m}} \right\|^2} < 1 + {\left\| {{{\mathbf{g}}_m}} \right\|^2}$. Therefore, using Lemma 1 results in a finite search space.\par

\begin{lemma}\label{lem:2}
For the solution of the optimization problem defined in~(\ref{eq:opt2}), we have either
\begin{equation}
\operatorname{sgn} \left( {{a_{ml}}} \right) = \operatorname{sgn} \left( {{g_{ml}}} \right), \textrm{ for all } l=1,\cdots,L,
\end{equation}
or
\begin{equation}
\operatorname{sgn} \left( {{a_{ml}}} \right) = -\operatorname{sgn} \left( {{g_{ml}}} \right), \textrm{ for all } l=1,\cdots,L,
\end{equation}
and if $g_{ml}=0$, then $a_{ml}=0$.
\end{lemma}\par

The operator $\operatorname{sgn}(\cdot)$ shows the sign function.\par
%
\begin{proof}
The only term in~(\ref{eq:Rga}) that depends on the sign of elements of ${\mathbf{a}}$ is the term ${{\left| {{{\mathbf{a}}^T}{{\mathbf{g}}_m}} \right|}^2}={\left| {\sum\nolimits_l {{a_{l}}{g_{ml}}} } \right|^2}$. For maximizing~(\ref{eq:Rga}), or equivalently (\ref{eq:opt2}), when the norm of ${\mathbf{a}}$ is fixed, this term has to be maximized. Thus, all the terms $a_{l}g_{ml} ,l=1,\cdots,L$ must have the same sign and this proves the lemma.\qed
\end{proof}\par
Lemma 2 suggests that  the search can be done over non-negative integers. Specifically, we replace the elements of ${\mathbf{g}}_m$ with their absolute values and solve the optimization problem~(\ref{eq:opt2}) by searching over non-negative  integers. The solution gives the absolute values of the elements of optimum ECV ${\mathbf{a}}_m$, where the signs of its elements are determined according to Lemma 2. It is noteworthy that both ${\mathbf{a}}_m$ and $-{\mathbf{a}}_m$ result in the same computation rates. In the rest of this paper, without loss of generality, we assume that all elements of ${\mathbf{a}}_m$ and ${\mathbf{g}}_m$ are non-negative.

\begin{lemma}\label{lem:3}
Any ECV ${\mathbf{a}}_m$ with $\gcd \left( {{a_{m1}}, \cdots ,{a_{mL}}} \right) > 1$ cannot be the solution of ~(\ref{eq:opt2}).
\end{lemma}\par
The operator $\gcd \left( {{a_{m1}}, \cdots ,{a_{mL}}}\right)$ shows the greatest common divisor of integers ${{a_{m1}}, \cdots ,{a_{mL}}}$ and we assume $\gcd (a,0)=a$.\par
\begin{proof}
Assume that $\gcd \left( {{a_{m1}}, \cdots ,{a_{mL}}}\right)=q>1$. If we define the integer vector ${\mathbf{b}}_m=\frac{1}{q}{\mathbf{a}}_m$, then we have
\begin{align}
{{\mathbf{b}}_m}^T{{\mathbf{G}}_m}{{\mathbf{b}}_m} = \frac{1}{{{q^2}}}{{\mathbf{a}}_m}^T{{\mathbf{G}}_m}{{\mathbf{a}}_m} < {{\mathbf{a}}_m}^T{{\mathbf{G}}_m}{{\mathbf{a}}_m},
\end{align}
and thus ${\mathbf{a}}_m$ cannot be the solution of~(\ref{eq:opt2}).\qed
\end{proof}\par
This lemma limits the search space of~(\ref{eq:opt2}) as well. It is sufficient to search over ECVs with $\gcd \left( {{a_{m1}}, \cdots ,{a_{mL}}}\right)=1$.\par

\begin{lemma}\label{lem:4}
Any ECV ${\mathbf{a}}_m$ with $\left\| {{{\mathbf{g}}_{min}}\left( {{{\mathbf{a}}_m}} \right)} \right\| > \left\| {{{\mathbf{g}}_m}} \right\|$ cannot be the solution of~(\ref{eq:opt2}), where 
\begin{equation}\label{eq:gmin_am}
\left\|{{\mathbf{g}}_{min}}\left( {{{\mathbf{a}}_m}} \right)  \right\| \triangleq
  \min_
{\substack{ {\mathbf{g}} \in {\mathbb{R}^L},{\mathbf{b}} \in {\mathbb{Z}^L} \\ 
{{\mathbf{a}}_m^T}{{\mathbf{G}}}{\mathbf{a}}_m \leq {{\mathbf{b}}^T}{{\mathbf{G}}}{\mathbf{b}} \\
{\left\| {\mathbf{b}} \right\|^2} \leqslant 1 + {\left\| {\mathbf{g}} \right\|^2} }}
\mspace{-10mu} \left\| {\mathbf{g}} \right\| ,
\end{equation}
and
\begin{equation}\label{eq:G_mat}
{{\mathbf{G}}} \triangleq \left( {{\mathbf{I}} - \frac{{{{\mathbf{g}}}{{\mathbf{g}}}^T}}{{1 + {{\left\| {{{\mathbf{g}}}} \right\|}^2}}}} \right).
\end{equation}
\end{lemma}\par
\begin{proof}
Assume that ECV ${\mathbf{a}}_m$ is the solution of~(\ref{eq:opt2}). Hence, it satisfies ${{\mathbf{a}}_m^T}{{\mathbf{G}}_m}{\mathbf{a}}_m \leq {{\mathbf{b}}^T}{{\mathbf{G}}_m}{\mathbf{b}}$ for all $\mathbf{b} \in {\mathbb{Z}^L}$. From Lemma 1, it is sufficient for ${\mathbf{a}}_m$ to satisfy ${{\mathbf{a}}_m^T}{{\mathbf{G}}_m}{\mathbf{a}}_m \leq {{\mathbf{b}}^T}{{\mathbf{G}}_m}{\mathbf{b}}$ only for all ${\left\| {\mathbf{b}} \right\|^2} \leqslant 1 + {\left\| {\mathbf{g}}_m \right\|^2}$. Therefore, ${\mathbf{g}}_m$ is a feasible point  of~(\ref{eq:gmin_am}) (since it satisfies both constraints of ~(\ref{eq:gmin_am})). As a result, for  the optimal solution of~(\ref{eq:gmin_am})  we have $\left\| {{{\mathbf{g}}_{min}}\left( {{{\mathbf{a}}_m}} \right)} \right\| \leq \left\| {{{\mathbf{g}}_m}} \right\|$.\qed
\end{proof}\par

Since ${\left\| {{{\mathbf{g}}_m}} \right\|^2}$ is the instantaneous sum SNR. It follows from Lemma~4 that the minimum instantaneous sum SNR for which an ECV  ${\mathbf{e}}_k$ can be selected as the solution of ~(\ref{eq:opt2}) equals $\left\| {{{\mathbf{g}}_{min}}\left( {{{\mathbf{e}}_k}} \right)} \right\|^2$. Optimization defined in~(\ref{eq:gmin_am}) is a mixed integer optimization problem and can be solved through numerical methods. Lemma 4 can significantly reduce the search space. To show its effect, we have computed the values of $\left\| {{{\mathbf{g}}_{min}}\left( {{{\mathbf{e}}_k}} \right)} \right\|$ for different ECVs ${\mathbf{e}}_k$ up to the instantaneous sum SNR about to 2000 (33 dB), in the case of $L=2$ sources. The results are sorted and provided in Tabel~\ref{tb:gmin}. As an example, in Fig.~\ref{fig:Regions}, the distance between point $F$ and the origin is equal to  $\left\| {{{\mathbf{g}}_{min}}\left( {{{\mathbf{e}}_k}} \right)} \right\|$ for ${\mathbf{e}}_k={[2,1]}^T$. Now, according to the Lemma 4, for a given ${\mathbf{g}}_m$, equivalent to a constant instantaneous sum SNR, only ECVs that satisfy $\left\| {{{\mathbf{g}}_{min}}\left( {{{\mathbf{a}}_m}} \right)} \right\| \leq
 \left\| {{{\mathbf{g}}_m}} \right\|$ are needed to be searched. These ECVs can be found from Tabel~\ref{tb:gmin}.  Note that Tabel~\ref{tb:gmin} is computed just once and thus adds no extra complexity to the method.This approach can be extended to fading channels, in which the instantaneous sum SNR has a certain probability distribution function, and with a high probability it is less than a threshold. This idea is the base of our proposed method which is described in the following subsection. \par

\begin{table}
\setlength{\tabcolsep}{4pt}
\centering
\caption{Computed $\left\| {{{\mathbf{g}}_{\mathrm{min}}}\left( {{{\mathbf{e}}_k}} \right)} \right\|$ for different ECVs, sorted in ascending order.}
\begin{tabular}{ccc|ccc}
\hline\hline
&&&&&\\[-2ex]
$k$ & ${\mathbf{e}}_k$ &  $\left\| {{{\mathbf{g}}_{min}}\left( {{{\mathbf{e}}_k}} \right)} \right\|^2$  &  $k$ & ${\mathbf{e}}_k$ &  $\left\| {{{\mathbf{g}}_{min}}\left( {{{\mathbf{e}}_k}} \right)} \right\|^2$ \\
\hline
&&&&&\\ [-2ex]
1 & $ [ 1 , 0 ] ^ T $ & 0 & 13 & $ [ 3 , 4 ] ^ T $ & 530.330 \\ 
2 & $ [ 0 , 1 ] ^ T $ & 0 & 14 & $ [ 5 , 1 ] ^ T $ & 626.000 \\
3 & $ [ 1 , 1 ] ^ T $ & 2 & 15 & $ [ 1 , 5 ] ^ T $ & 626.000 \\
4 & $ [ 2 , 1 ] ^ T $ & 18.282 & 16 & $ [ 5 , 2 ] ^ T $ & 642.334 \\
5 & $ [ 1 , 2 ] ^ T $ & 18.282 & 17 & $ [ 2 , 5 ] ^ T $ & 642.334 \\
6 & $ [ 3 , 1 ] ^ T $ & 82.321 & 18 & $ [ 5 , 3 ] ^ T $ & 898.333 \\
7 & $ [ 1 , 3 ] ^ T $ & 82.321 & 19 & $ [ 3 , 5 ] ^ T $ & 898.333 \\
8 & $ [ 3 , 2 ] ^ T $ & 130.325 & 20 & $ [ 6 , 1 ] ^ T $ & 1297.001 \\
9 & $ [ 2 , 3 ] ^ T $ & 130.325 & 21 & $ [ 1 , 6 ] ^ T $ & 1297.001 \\
10 & $ [ 4 , 1 ] ^ T $ & 256.996 & 22 & $ [ 5 , 4 ] ^ T $ & 1521.999 \\
11 & $ [ 1 , 4 ] ^ T $ & 256.996 & 23 & $ [ 4 , 5 ] ^ T $ & 1521.999 \\
12 & $ [ 4 , 3 ] ^ T $ & 530.330 & 24 & $ [ 7 , 2 ] ^ T $ & 2130.330 \\
\end{tabular}
\label{tb:gmin}
\end{table}
 
All of the provided lemmas reduce the complexity of optimum compute-and-forward method by limiting the search space of ~(\ref{eq:opt2}). As an example, consider two vectors ${\mathbf{g}}_1$ and ${\mathbf{g}}_2$ with ${\left\| {{{\mathbf{g}}_1}} \right\|^2}=100$ (i.e. 20 dB instantaneous sum SNR) and ${\left\| {{{\mathbf{g}}_2}} \right\|^2}=1000$ (i.e. 30 dB instantaneous sum SNR), respectively. Lemma 1 limits the size of search space to 317 and 3141 vectors, respectively. Lemma 2 (in conjunction with Lemma 1) decreases it to  89 and 818, respectively. The size reduces to 49 and 479, respectively, by using Lemma 3.  Finally, by exploiting all of the four lemmas, the search set only contains 7 and 23 vectors, respectively. However, as it is observed, by increasing the SNR, the search set and complexity grow considerably.\par
The following lemma facilitates 
the completion of Tabel~\ref{tb:gmin}.
\begin{lemma}\label{lem:5}
For Any permutation ${{\mathbf{\tilde e}}_k}$ of ECV ${\mathbf{e}}_k$, we have $\left\| {{{\mathbf{g}}_{min}}\left( {{{{\mathbf{\tilde e}}}_k}} \right)} \right\| = \left\| {{{\mathbf{g}}_{min}}\left( {{{\mathbf{e}}_k}} \right)} \right\|$.
\end{lemma}
\begin{proof}
Let ${{\mathbf{\tilde e}}_k}$ be a permutation of ${\mathbf{e}}_k$, then there exists a permutation matrix $\mathbf{U}$ such that ${{{\mathbf{\tilde e}}}_k} = {\mathbf{U}}{{\mathbf{e}}_k}$. According to the properties of a permutation matrix, we have that $\left\| {{{{\mathbf{\tilde e}}}_k}} \right\| = \left\| {{\mathbf{U}}{{\mathbf{e}}_k}} \right\| = \left\| {{{\mathbf{e}}_k}} \right\|$ and ${{\mathbf{U}}^{ - 1}} = {{\mathbf{U}}^T}$. From~(\ref{eq:gmin_am}) and (\ref{eq:G_mat}), we can write

\begin{equation}\label{eq:gmin_am2}
\left\|{{\mathbf{g}}_{min}}\left( {{{\mathbf{\tilde e}}_k}} \right)  \right\| \triangleq
  \min_
{\substack{ {\mathbf{\tilde g}} \in {\mathbb{R}^L},{\mathbf{\tilde b}} \in {\mathbb{Z}^L} \\ 
{{\mathbf{\tilde e}}_k^T}{{\mathbf{\tilde G}}}{\mathbf{\tilde e}}_k \leq {{\mathbf{\tilde b}}^T}{{\mathbf{\tilde G}}}{\mathbf{\tilde b}} \\
{\left\| {\mathbf{\tilde b}} \right\|^2} \leqslant 1 + {\left\| {\mathbf{\tilde g}} \right\|^2} }}
\mspace{-10mu} \left\| {\mathbf{\tilde g}} \right\| ,
\end{equation}
where
\begin{equation}\label{eq:G_mat_tilde}
{{\mathbf{\tilde G}}} \triangleq \left( {{\mathbf{I}} - \frac{{{{\mathbf{\tilde g}}}{{\mathbf{\tilde g}}}^T}}{{1 + {{\left\| {{{\mathbf{\tilde g}}}} \right\|}^2}}}} \right).
\end{equation}
By defining ${\mathbf{g}} \triangleq {{\mathbf{U}}^{ - 1}}{\mathbf{\tilde g}} = {{\mathbf{U}}^T}{\mathbf{\tilde g}} \in {\mathbb{R}^L}$ and ${\mathbf{b}} \triangleq {{\mathbf{U}}^{ - 1}}{\mathbf{\tilde b}} = {{\mathbf{U}}^T}{\mathbf{\tilde b}} \in {\mathbb{Z}^L}$, which leads to ${\mathbf{\tilde g}} = {\mathbf{Ug}}$, ${\mathbf{\tilde b}} = {\mathbf{Ub}}$, $\left\| {{\mathbf{\tilde g}}} \right\| = \left\| {\mathbf{g}} \right\|$ and $\| {{\mathbf{\tilde b}}} \| = \left\| {\mathbf{b}} \right\|$, we have
\begin{equation}
\begin{aligned}
{\mathbf{\tilde e}}_k^T{\mathbf{\tilde G}}{{{\mathbf{\tilde e}}}_k} &= {\left\| {{{{\mathbf{\tilde e}}}_k}} \right\|^2} - \frac{{{{\left| {{{{\mathbf{\tilde g}}}^T}{{{\mathbf{\tilde e}}}_k}} \right|}^2}}}{{1 + {{\left\| {{\mathbf{\tilde g}}} \right\|}^2}}}\\
& = {\left\| {{{\mathbf{e}}_k}} \right\|^2} - \frac{{{{\left| {{{\mathbf{g}}^T}{{\mathbf{U}}^T}{\mathbf{U}}{{\mathbf{e}}_k}} \right|}^2}}}{{1 + {{\left\| {{\mathbf{Ug}}} \right\|}^2}}} \\
&= {\left\| {{{\mathbf{e}}_k}} \right\|^2} - \frac{{{{\left| {{{\mathbf{g}}^T}{{\mathbf{e}}_k}} \right|}^2}}}{{1 + {{\left\| {\mathbf{g}} \right\|}^2}}}\\
& = {{\mathbf{e}}_k}^T{\mathbf{G}}{{\mathbf{e}}_k}
\end{aligned}
\end{equation}
and, in a similar way, ${{{\mathbf{\tilde b}}}^T}{\mathbf{\tilde G\tilde b}} = {{\mathbf{b}}^T}{\mathbf{Gb}}$. 
This yields that minimization in~(\ref{eq:gmin_am2}) is the same as the one in~(\ref{eq:gmin_am}) and this proves the lemma.\qed
\end{proof}\par
As an example, ${\mathbf{e}}_{25}={[2,7]}^T$ is a permutation of ${\mathbf{e}}_{24}={[7,2]}^T$, in Table~\ref{tb:gmin}, and hence we have $\left\| {{{\mathbf{g}}_{min}}\left( {{{\mathbf{e}}_{25}}} \right)} \right\| = \left\| {{{\mathbf{g}}_{min}}\left( {{{\mathbf{e}}_{24}}} \right)} \right\|$.

\subsection{Simplified Compute-and-Forward (Proposed Scheme)}
As mentioned in the previous subsection, the idea of our method is based on limiting the search set of the optimization problem. To describe our simplified scheme, in Fig.~\ref{fig:PrSel} we have plotted the probability of selecting different ECVs in a relay, employing optimum CMF method, versus average SNR , as an example, for the case of equal transmission powers, i.e. $P_1=P_2=P$, in Rayleigh fading channel. From this figure, for SNRs less than 8 dB, the ECV of the optimum CMF is in the set $\left\{ {{{[1,0]}^T},{{[0,1]}^T},{{[1,1]}^T}} \right\}$  with a high probability. Likewise, for SNRs less than 16 dB, the ECV of the optimum CMF scheme is in the set $\left\{ {{{[1,0]}^T},{{[0,1]}^T},{{[1,1]}^T},{{[2,1]}^T},{{[1,2]}^T}} \right\}$ with a high probability. These results are consistent with the results provided in Table~\ref{tb:gmin} as well. When the average SNR (of each user) is less than 8 dB, the average instantaneous sum SNR, i.e. ${\left\| {{{\mathbf{g}}_m}} \right\|^2}$, is less than 12.7. Thus, according to the Table~\ref{tb:gmin}, the first three ECVs can be selected as an optimum solution with a high probability. In a similar manner, for average SNR (of each user) less than 16 dB, the average instantaneous sum SNR is less than 79.6. Thus, based on the Table~\ref{tb:gmin}, the first five ECVs cab be selected with a high probability. That is Table~\ref{tb:gmin} confirms Fig.~\ref{fig:PrSel}.\par

From the above observations, which is true for the other examples considered as well, we propose a simplified CMF, in which the search set for finding the proper values of ECV is finite and its size can be properly selected based on the channel received SNR (using Table~\ref{tb:gmin}). That is
\begin{equation}\label{eq:Opt_sim}
 {{\mathbf{a}}_m} = \arg \min_{\substack{{\mathbf{a}} \in {\mathbf{S}}_K\\{\mathbf{a}} \ne {\mathbf{0}}}} {{\mathbf{a}}^T}{{\mathbf{G}}_m}{\mathbf{a}},
\end{equation}
where ${\mathbf{S}}_K$ is the ECV search set including $K$ vectors. For instance, for the example considered, we can select

\begin{equation}
\begin{aligned}
  {{\mathbf{S}}_2}& = \left\{ {{{\mathbf{e}}_1},{{\mathbf{e}}_2}} \right\},\\
  {{\mathbf{S}}_3}& = \left\{ {{{\mathbf{e}}_1},{{\mathbf{e}}_2},{{\mathbf{e}}_3}} \right\},\\
  {{\mathbf{S}}_5}& = \left\{ {{{\mathbf{e}}_1},{{\mathbf{e}}_2},{{\mathbf{e}}_3},{{\mathbf{e}}_4},{{\mathbf{e}}_5}} \right\},
\end{aligned}
\end{equation}
and so on, where ${\mathbf{e}}_1={[1,0]}^T$, ${\mathbf{e}}_2={[0,1]}^T$, ${\mathbf{e}}_3={[1,1]}^T$, ${\mathbf{e}}_4={[2,1]}^T$, and ${\mathbf{e}}_5={[1,2]}^T$ are derived from Table~\ref{tb:gmin}.\par
We denote the simplified CMF, with the ECV search set ${\mathbf{S}}_K$, by CMF($K$). In fact, CMF($\infty$) is equivalent to the optimum CMF. In CMF($K$), the complexity of finding ECV is in the order of $KM$, which is significantly less than that of the optimum CMF.\par
It is noteworthy to mention that since CMF($K$) limits the search region of ECV selection, the outage Probability of CMF($K$)  upper bounds that of the optimum CMF, for all values of $K$.\par

\begin{figure}[t!]
\renewcommand{\figurename}{Fig.}
\centering
\includegraphics[width = \columnwidth]{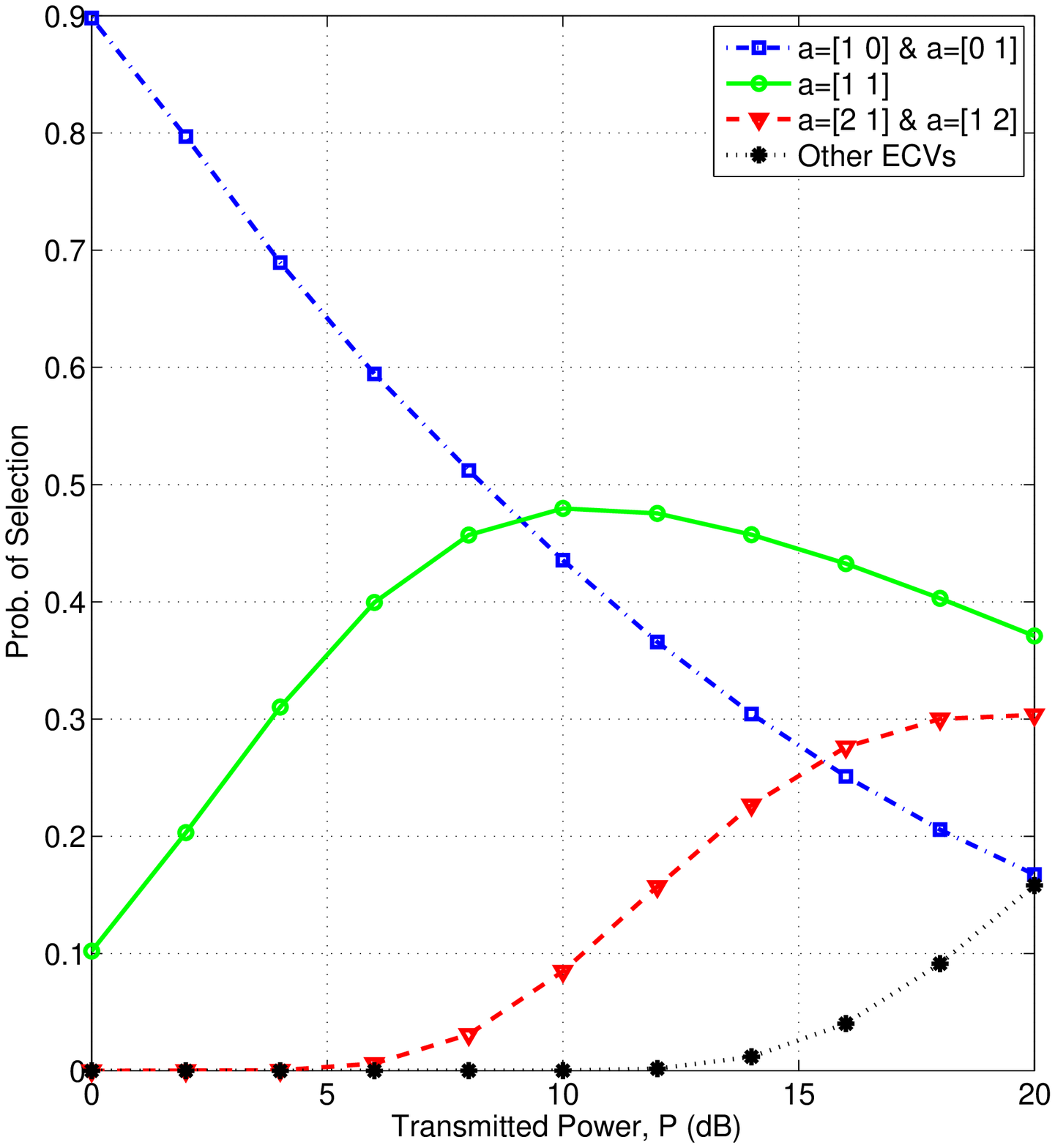}
\caption{Probability of selecting different ECVs in a relay for optimum CMF method versus average SNR ($P_1=P_2=P$).}
\label{fig:PrSel}
\end{figure}

\section{PERFORMANCE ANALYSIS} \label{sec:analysis}
In this section, we present the exact end-to-end performance analysis of CMF($K$) in terms of the outage probability, which provides an upper bound for the outage probability of the optimum CMF method.\par

\subsection{Probability of selecting an ECV in a relay}
From (3) and (6), CMF($K$) selects ECV ${\mathbf{e}}_k$, for $k=1,2,…,K$ , if and only if $R({{\mathbf{g}}_m},{{\mathbf{e}}_k}) \geqslant R({{\mathbf{g}}_m},{{\mathbf{e}}_l}),\forall l \ne k$. Thus the probability of selecting ${\mathbf{e}}_k$ is equal to the Probability that ${\mathbf{g}}_m$ is in the region $D_k$ defined as
\begin{equation}\label{eq:Dk_def}
\begin{aligned}
{D_k} &\triangleq \left\{ {{\mathbf{g}}|R({\mathbf{g}},{{\mathbf{e}}_k}) \geqslant R({\mathbf{g}},{{\mathbf{e}}_l}),\forall l \ne k} \right\}\\
&= \left\{ {{\mathbf{g}}|{{\mathbf{e}}_k}^T{\mathbf{G}}{{\mathbf{e}}_k} \leqslant {{\mathbf{e}}_l}^T{\mathbf{G}}{{\mathbf{e}}_l},\forall l \ne k} \right\},k = 1, \ldots ,K,
\end{aligned}
\end{equation}
where $\mathbf{G}$ is defined in~(\ref{eq:G_mat}). This probability can be computed as
\begin{equation}\label{eq:PkSelDef}
\begin{aligned}
P_k^{Sel} &\triangleq \Pr \left\{ {{{\mathbf{a}}_m} = {{\mathbf{e}}_k}} \right\}=\Pr \left\{ {{{\mathbf{g}}_m} \in {D_k}} \right\}\\
&= \iint\limits_{{{\mathbf{g}}_m} \in {D_k}} \mspace{-10mu} {f({g_{m1}},{g_{m2}}) \, d{g_{m1}}d{g_{m2}}} \, ,k = 1, \ldots ,K,
\end{aligned}
\end{equation}
where $f({g_{m1}},{g_{m2}})$ is the joint probability distribution of $g_{m1}$ and $g_{m2}$. Sine $g_{m1}$ and $g_{m2}$ are independent, we have
\begin{equation}
f({g_{m1}},{g_{m2}}) = \frac{{2{g_{m1}}}}{{{P_1}}}{e^{ - \frac{{{g_{m1}}^2}}{{{P_1}}}}}.\frac{{2{g_{m2}}}}{{{P_2}}}{e^{ - \frac{{{g_{m2}}^2}}{{{P_2}}}}}.
\end{equation}\par
Region $D_k$ is calculated from~(\ref{eq:Dk_def}) by computing the borders ${{\mathbf{e}}_k}^T{\mathbf{G}}{{\mathbf{e}}_k} = {{\mathbf{e}}_l}^T{\mathbf{G}}{{\mathbf{e}}_l},\forall l \ne k$, which are in general, in terms of $g_1$ and $g_2$, hyperbolas.   As an example, Fig.~\ref{fig:Regions} shows the selection regions $D_k$ for CMF(5). Note that determining the regions $D_k, k=1,\cdots,K$ is not required for finding  ${\mathbf{a}}_m$, just required for the performance analysis. ECV , i.e. ${\mathbf{a}}_m$, is found simply by calculating and comparing the terms $\left\{ {{{\mathbf{e}}_k}^T{{\mathbf{G}}_m}{{\mathbf{e}}_k},k = 1, \cdots ,K} \right\}$ according to~(\ref{eq:Opt_sim}).
\par

\begin{figure}[t]
\renewcommand{\figurename}{Fig.}
\centering
\includegraphics[width = \columnwidth]{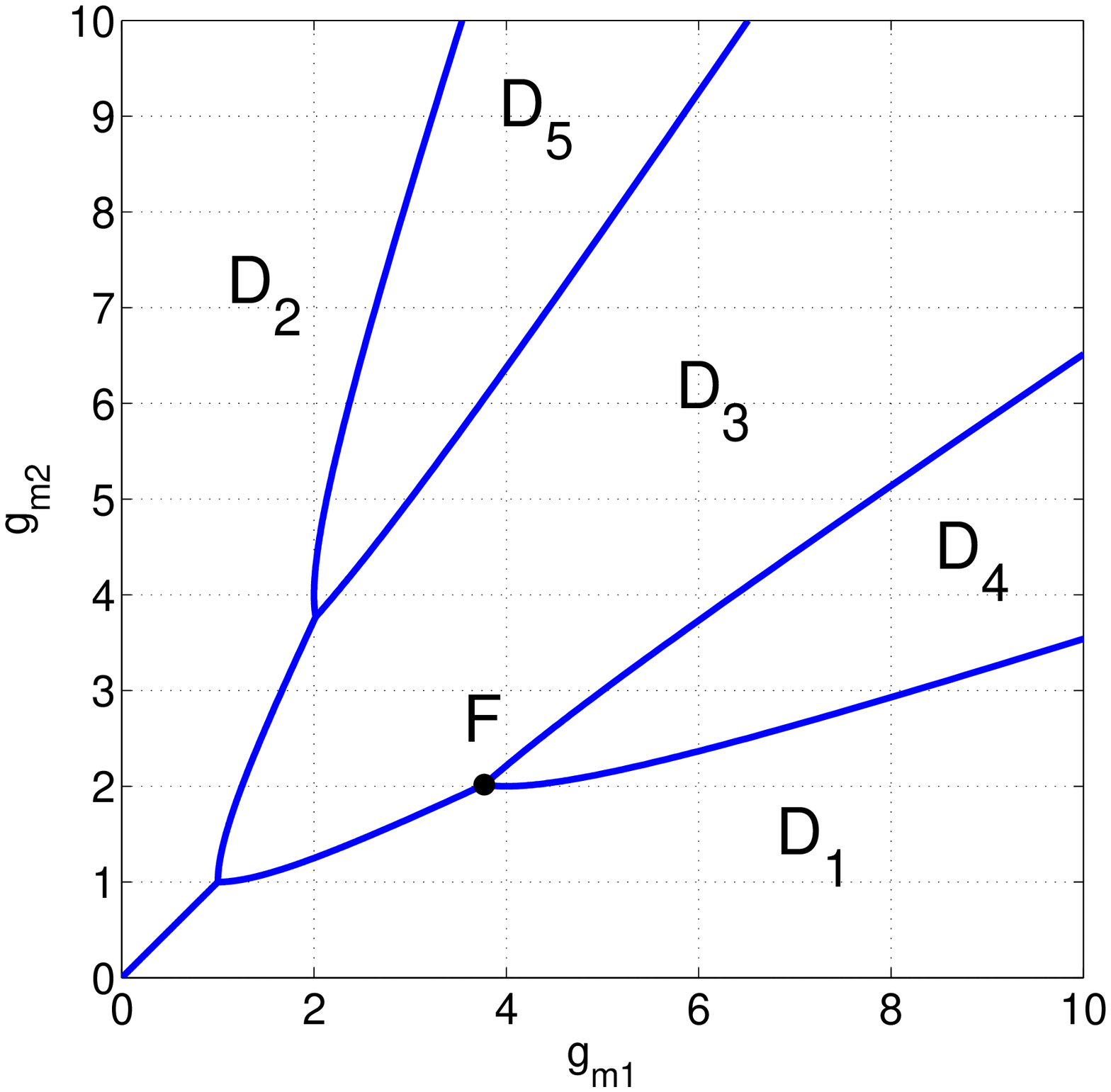}
\caption{Selection regions Dk for CMF(5).}
\label{fig:Regions}
\end{figure}

\subsection{The outage probability of each relay}
The outage Probability of a relay is the probability that the relay computation rate, i.e. $R({{\mathbf{g}}_m},{{\mathbf{a}}_m})$, is less than a given target rate $R_t$, and can be expressed as
\begin{equation}
P_{relay}^{Out} = \Pr \left\{ {R({{\mathbf{g}}_m},{{\mathbf{a}}_m}) < \,{R_t}} \right\}.
\end{equation}\par
Using the law of total probability, we can write
\begin{equation}\label{eq:PoutRelay2}
\footnotesize
\begin{aligned}
  P_{relay}^{Out} &=\sum\limits_{k = 1}^K {\Pr \left\{ {{{\mathbf{g}}_m} \in {D_k}} \right\}.\Pr \left\{ {R({{\mathbf{g}}_m},{{\mathbf{a}}_m}) < {R_t}|{{\mathbf{g}}_m} \in {D_k}} \right\}}\\
&\mathop  = \limits^{(a)}  \sum\limits_{k = 1}^K {P_k^{Sel}.\Pr \left\{ {R({{\mathbf{g}}_m},{{\mathbf{e}}_k}) < {R_t},\,{{\mathbf{g}}_m}\in {D_k}} \right\}},
\end{aligned} 
\end{equation}
where $P_k^{Sel}$ is defined in (\ref{eq:PkSelDef}) and $(a)$ holds since if ${\mathbf{g}}_m \in {D_k}$ then we have ${{\mathbf{a}}_m}={{\mathbf{e}}_k}$. We define conditional outage probability as
\begin{equation}\label{eq:Pout_ek_Def}
{P_{Out|{{\mathbf{e}}_k}}} \triangleq \Pr \left\{ {R({{\mathbf{g}}_m},{{\mathbf{e}}_k}) < {R_t}|{{\mathbf{g}}_m} \in {D_k}} \right\}.
\end{equation}\par
Using (\ref{eq:PoutRelay2}) and (\ref{eq:Pout_ek_Def}) yields
\begin{equation}
P_{relay}^{Out} = \sum\limits_{k = 1}^K {P_k^{Sel}.{P_{Out|{{\mathbf{e}}_k}}}}.
\end{equation}\par
By defining region
\begin{equation}
\begin{aligned}
{O_k} &\triangleq \left\{ {{\mathbf{g}}|R({\mathbf{g}},{{\mathbf{e}}_k}) < {R_t}} \right\}\\
&= \left\{ {{\mathbf{g}}|{{\mathbf{e}}_k}^T{\mathbf{G}}{{\mathbf{e}}_k} > \,\frac{1}{{{2^{2{R_t}}}}}} \right\}\,,k = 1, \ldots ,K, 
\end{aligned} 
\end{equation}
the conditional outage probability  ${P_{Out|{{\mathbf{e}}_k}}}$ can be easily calculated as
\begin{equation}\label{eq:Pout_ek}
\begin{aligned}
  {P_{Out|{{\mathbf{e}}_k}}} &= \Pr \left\{ {{{\mathbf{g}}_m} \in {O_k}|{{\mathbf{g}}_m} \in {D_k}} \right\}\\
&= \frac{{\Pr \left\{ {{{\mathbf{g}}_m} \in {O_k} \cap {D_k}} \right\}}}{{\Pr \left\{ {{{\mathbf{g}}_m} \in {D_k}} \right\}}} \\
 &= \frac{{\iint\limits_{{{\mathbf{g}}_m} \in {O_k} \cap {D_k}} \mspace{-25mu} {f({g_{m1}},{g_{m2}}) \, d{g_{m1}}d{g_{m2}}}}}{{P_k^{Sel}}}\,. 
\end{aligned} 
\end{equation}\par
Region $O_k$ is the area within the border ${{{\mathbf{e}}_k}^T{\mathbf{G}}{{\mathbf{e}}_k} = {{{2^{-2{R_t}}}}}}$, which can be easily shown that, in terms of $g_1$ and $g_2$, it is a hyperbola.

\subsection{Probability of rank failure}
As stated before, the destination after receiving $M$ equations from the relays attempts to find and solve the two best independent equations, i.e. the two equations with the highest computation rates, to recover the both messages . Since the relays select their ECVs independently, the equations are not necessarily independent. If the destination fails to find two independent equations, a rank failure occurs and the destination cannot recover both transmitted messages. This event causes an outage as well.\par
A rank failure occurs if all the $M$ equations are linearly dependent. In a two-dimensional space, two vectors $\mathbf{a}$ and $\mathbf{b}$ are linearly dependent only when they are in the same direction, i.e. ${\mathbf{a}}=\alpha {\mathbf{b}}, \alpha \in \mathbb{R}$. From Lemma~\ref{lem:3} and the fact that ECVs are integer vectors, two ECVs are linearly dependent only when they are equal. Therefore, a rank failure occurs at the destination if all the received equations are the same. Thus for CMF($K$), we can compute the probability of rank failure as
 \begin{equation}
\begin{aligned}
  P_{fail} &= \sum\limits_{k = 1}^K {\Pr \left\{ {{{\mathbf{a}}_m} = {{\mathbf{e}}_k}\,,\forall m = 1, \ldots ,M} \right\}} \\
  &\mathop  = \limits^{(a)} \sum\limits_{k = 1}^K {\Pr {{\left\{ {{{\mathbf{a}}_m} = {{\mathbf{e}}_k}\,} \right\}}^M}} \\
   &= \sum\limits_{k = 1}^K {\left( {P_k^{Sel}} \right)}
\end{aligned} 
\end{equation}
where $P_k^{Sel}$ is defined in~(\ref{eq:PkSelDef}) and $(a)$ follows from the fact that relays select their ECVs independently.\par
It is useful to find a lower bound for the overall outage probability of the system. As stated above, a rank failure, with probability one causes an outage in the system. Hence, we can write
\begin{equation}\label{eq:pout_fail}
\begin{aligned}
P_{out}^{sys} &= {{P} _{fail}} + \left( {1 - {{P} _{fail}}} \right).P_{out|No\,fail}^{sys}\\
&\geq {P} _{fail}
\end{aligned}
\end{equation}
where $P_{out}^{sys}$ is the overall outage probability of the system and $P_{out|No\,fail}^{sys}$ is the overall outage probability of the system conditioned on no rank failure. A direct result of the above equation is that $P_{fail}$ lower bounds the overall outage probability of the system. \par
It is noteworthy that~(\ref{eq:pout_fail}) is true for the optimum CMF method as well (with a different probability of rank failure $P_{fail}$). Hence the probability of rank failure imposes a lower bound on the outage probability of the optimum CMF method as well. \par

\subsection{The outage probability of the system}
The destination receives $M$ equations as well as their corresponding ECVs $\left\{ {{{\mathbf{a}}_m}} \right\}_{m = 1}^M$ and the computation rates $\left\{ {{R_m}} \right\}_{m = 1}^M$. Then, the receiver selects two independent equations with the highest rates. An outage occurs at the destination if either the minimum rate of these two selected equations is less than the given target rate $R_t$ or the destination cannot find two independent ECVs among $M$ received ECVs. Hence, the system outage probability can be expressed as
\begin{equation}\label{eq:PoutSys1}
P_{out}^{sys} = \Pr \left\{{
 {\max_
 {\substack{
 1 \leqslant i,j \leqslant M,\,\,i \ne j,\\\det \left( {\left[ {{{\mathbf{a}}_i},{{\mathbf{a}}_j}} \right]} \right) \ne 0
 }}
  {\min \left( {{R_i}\,,{R_j}} \right)}
  }
 \quad< 
 {R_t} 
   }\right\}.
\end{equation}\par
The probability in the above equation includes the event of rank failure in the destination, i.e. the case that destination cannot find two independent equations. In this case, the search space of maximization in~(\ref{eq:PoutSys1}) is an empty set and hence considered as an outage event.\par
To compute the probability in ~(\ref{eq:PoutSys1}) for the proposed scheme of CMF(K), we present a strategy for selecting two independent equations (we have assumed that $L=2$) with the highest minimum rates, as follows: All the $M$ received equations are divided into $K$ sets, each having the same ECV. Please note that some of these sets can be empty. Then, the equation with the highest rate is selected from the non-empty sets. As a result, at most $K$ independent equations, each from a non-empty set, are selected. Finally, the two equation with the highest rates among these independent equations are the desired ones. In the following, we calculate~(\ref{eq:PoutSys1}) based on this strategy. \par
Define ${\mathbf{T}}_k , k=1,…,K$, as the set of received equations that their ECVs are ${\mathbf{e}}_k$. The size of the set ${\mathbf{T}}_k$ is denoted by $n_k$. From the definition of $P_k^{Sel}$ in (\ref{eq:PkSelDef}), the probability of having the size vector of $\left({n_1},{n_2},\cdots,{n_K}\right)$ is easily computed as
\begin{equation}\label{eq:P_n1_to_nk}
\Pr \left( {{n_1}, \cdots ,{n_K}} \right) = \left( {\begin{array}{*{20}{c}} 
  M \\ 
  {{n_1}, \cdots ,{n_K}} 
\end{array}} \right) \cdot \prod\limits_{k = 1}^K {{{\left( {P_k^{Sel}} \right)}^{{n_k}}}} ,
\end{equation}
where
\begin{equation}
\left( {\begin{array}{*{20}{c}}
  M \\ 
  {{n_1}, \cdots ,{n_K}} 
\end{array}} \right) \triangleq \frac{{M!}}{{{n_1}! \cdots {n_K}!}}.
\end{equation}\par
Thus, using the law of total probability, we drive
\begin{equation}\label{eq:PoutSys2}
P_{out}^{sys} = \sum\limits_{\substack{ 
  \left( {{n_1}, \cdots ,{n_K}} \right) \\ 
  \forall k:0 \leqslant {n_k} \leqslant M \\ 
  \sum\nolimits_{k = 1}^K {{n_k}}  = M }}
  {\Pr \left( {{n_1}, \cdots ,{n_K}} \right) \cdot P_{out|\left( {{n_1}, \cdots ,{n_K}} \right)}^{sys}} \,,
\end{equation}
where $P_{out|\left( {{n_1}, \cdots ,{n_K}} \right)}^{sys}$ is the outage probability given the size vector $\left({n_1},{n_2},\cdots,{n_K}\right)$. To compute the conditional outage , let ${R'_k}$ denotes the maximum computation rate of all equations in ${\mathbf{T}}_k$ . If ${\mathbf{T}}_k$ is an empty set we set ${R'_k} = 0$. Moreover, let the set $\left\{ {{R_{k,1}},{R_{k,2}}, \cdots ,{R_{k,{n_k}}}} \right\}$ include all computation rates of equations in ${\mathbf{T}}_k$. Since all relays have the same conditions, all the computation rates are \textit{i.i.d.}, and we can calculate
\begin{equation}
\begin{aligned}
  \Pr &\left\{ {{{R'}_k} < {R_t}} \right\} \\
 &= \Pr \left\{ {\max \left( {{R_{k,1}}, \cdots ,{R_{k,{n_k}}}} \right) < {R_t}} \right\} \\
 &= \Pr \left\{ {\bigcap\limits_{r = 1}^{{n_k}} {{R_{k,r}} < {R_t}} } \right\} \\
 &= \prod\limits_{r = 1}^{{n_k}} {\Pr \left\{ {{R_{k,r}} < {R_t}} \right\}}  \\
 &= {\left( {{P_{Out|{{\mathbf{e}}_k}}}} \right)^{{n_k}}},
\end{aligned}
\end{equation}•
where ${P_{Out|{{\mathbf{e}}_k}}}$  is given in (\ref{eq:Pout_ek}).\par
The conditional probability in (\ref{eq:PoutSys2}) can be then computed as (\ref{eq:PoutSys_n1_to_nk}), where ${\max _{(2)}} \left( \cdot \right)$  represents the second maximum operator and $(a)$ follows from the fact that all the computation rates are \textit{i.i.d.}.\par

\begin{equation}\label{eq:PoutSys_n1_to_nk}
\begin{split}
  P_{out|\left( {{n_1}, \cdots ,{n_K}} \right)}^{sys} &= \Pr \left\{ {{{\max }_{(2)}}\left( {{{R'}_1},{{R'}_2}, \cdots ,{{R'}_K}} \right) < {R_t}} \right\} \\
  &= \Pr \left\{ \left( {\bigcap\limits_{k = 1}^K {{{R'}_k} < {R_t}} } \right) \bigcup {\left( {\bigcup\limits_{k = 1}^K {\left( {{{R'}_k} > {R_t}\bigcap\limits_{j = 1\,,j \ne k}^K {{{R'}_j} < {R_t}} } \right)} } \right)}  \right\} \\
 &= \Pr \left\{ {\bigcap\limits_{k = 1}^K {{{R'}_k} < {R_t}} \,} \right\} + \sum\limits_{k = 1}^K {\Pr \left\{ {\,{{R'}_k} > {R_t}\bigcap\limits_{j = 1\,,j \ne k}^K {{{R'}_j} < {R_t}} } \right\}} \\
 &\mathop  = \limits^{(a)}  \prod\limits_{k = 1}^K {\Pr \left\{ {{{R'}_k} < {R_t}\,} \right\}}  + \sum\limits_{k = 1}^K {\left( {\Pr \left\{ {{{R'}_k} > {R_t}\,} \right\} \cdot \prod\limits_{j = 1\,,j \ne k}^K {\Pr \left\{ {{{R'}_j} < {R_t}\,} \right\}} } \right)}\\
 &= \prod\limits_{k = 1}^K {{{\left( {{P_{Out|{{\mathbf{e}}_k}}}} \right)}^{{n_k}}}}  + \sum\limits_{k = 1}^K {\left( {\left( {1 - {{\left( {{P_{Out|{{\mathbf{e}}_k}}}} \right)}^{{n_k}}}} \right) \cdot \prod\limits_{j = 1\,,j \ne k}^K {{{\left( {{P_{Out|{{\mathbf{e}}_j}}}} \right)}^{{n_j}}}} } \right)\,} 
\end{split}
\end{equation}


Finally, by substitution of (\ref{eq:P_n1_to_nk}) and (\ref{eq:PoutSys_n1_to_nk}) in (\ref{eq:PoutSys2}), the overall system outage probability is derived as (\ref{eq:PoutSys_f}), where  $P_k^{Sel}$ and ${P_{Out|{{\mathbf{e}}_k}}}$  are  given in (\ref{eq:PkSelDef}) and (\ref{eq:Pout_ek}), respectively.
\begin{multline}\label{eq:PoutSys_f}
  P_{out}^{sys} =
\sum\limits_{\substack{\left( {{n_1}, \cdots ,{n_K}} \right) \\ 
\forall k:0 \leqslant {n_k} \leqslant M \\ 
  \sum\nolimits_{k = 1}^K {{n_k}}  = M }}
  {\left( {\begin{array}{*{20}{c}}
  M \\ 
  {{n_1}, \cdots ,{n_K}} 
\end{array}} \right) \cdot \prod\limits_{k = 1}^K {{{\left( {P_k^{Sel}} \right)}^{{n_k}}}} }\\
  \cdot \left( {\prod\limits_{k = 1}^K {{{\left( {{P_{Out|{{\mathbf{e}}_k}}}} \right)}^{{n_k}}}}  + \sum\limits_{k = 1}^K {\left( {\left( {1 - {{\left( {{P_{Out|{{\mathbf{e}}_k}}}} \right)}^{{n_k}}}} \right) \cdot \prod\limits_{j = 1,j \ne k}^K {{{\left( {{P_{Out|{{\mathbf{e}}_j}}}} \right)}^{{n_j}}}} } \right)} } \right)
\end{multline}

\section{SIMULATION RESULTS} \label{sec:simulation}
In this section, computer simulations and analytical results are provided to study and compare optimum and simplified methods. In simulations, equal powers $P_1=P_2=P$ and target rate $R_t=0.5$ are assumed.\par
Fig.~\ref{fig:det} shows the probability of rank failure in the destination for optimum CMF, CMF(3), and CMF(5) methods versus the average SNR. Moreover, two cases of  $M=2$ and $M=6$ are compared. As it is observed, by increasing the number of relays ($M$), probability of rank failure decreases significantly. The reason is that by the increase of $M$, the destination receives more equations and hence it can find two independent equation with a higher probability. Generally, the optimum CMF shows a lower probability of rank failure compared with CMF(3) and CMF(5). Similarly, the probability of rank failure of CMF(5) is less than that of CMF(3). This is due to the fact that in CMF(5), the search set has more elements  than in CMF(3). Hence, the probability of selecting the same ECVs in different relays decreases.\par
\begin{figure}[tb]
\renewcommand{\figurename}{Fig.}
\centering
\includegraphics[width = \columnwidth]{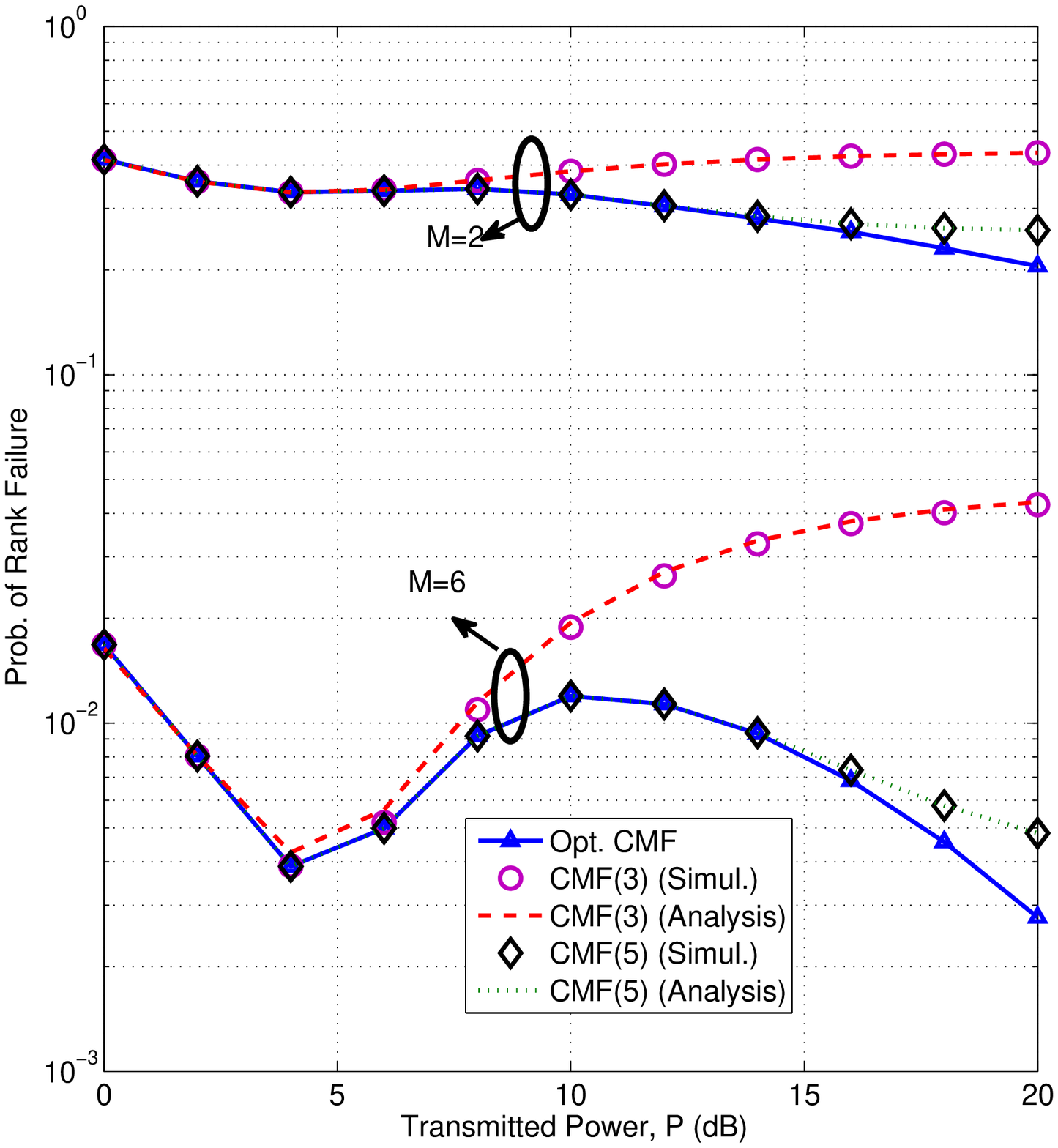}
\caption{Probability of rank failure in destination for optimum CMF, CMF(3), and CMF(5) versus average SNR ($M=2, 6$).}
\label{fig:det}
\end{figure}

Fig.~\ref{fig:PoutSys} compares the systems outage probability of optimum CMF, CMF(3), and CMF(5) versus average SNR in the cases of $M=2$ and $M=6$ relays, based on computer simulations. Analytical results are also provided for our proposed simplified schemes for the comparisons. First, as can be realized, simulation results for CMF(3) and CMF(5) well coincide with the related analytical results. From this  figure, CMF(3) and CMF(5) outage curves provide upper bounds on optimum CMF outage curves, which are closely tight in SNRs less than 6 dB and 16 dB, respectively. As it is observed, CMF(5) performs near the optimum CMF method up to a higher SNR threshold and provides a wider valid approximation region than CMF(3). \par
\begin{figure}[tb]
\renewcommand{\figurename}{Fig.}
\centering
\includegraphics[width = \columnwidth]{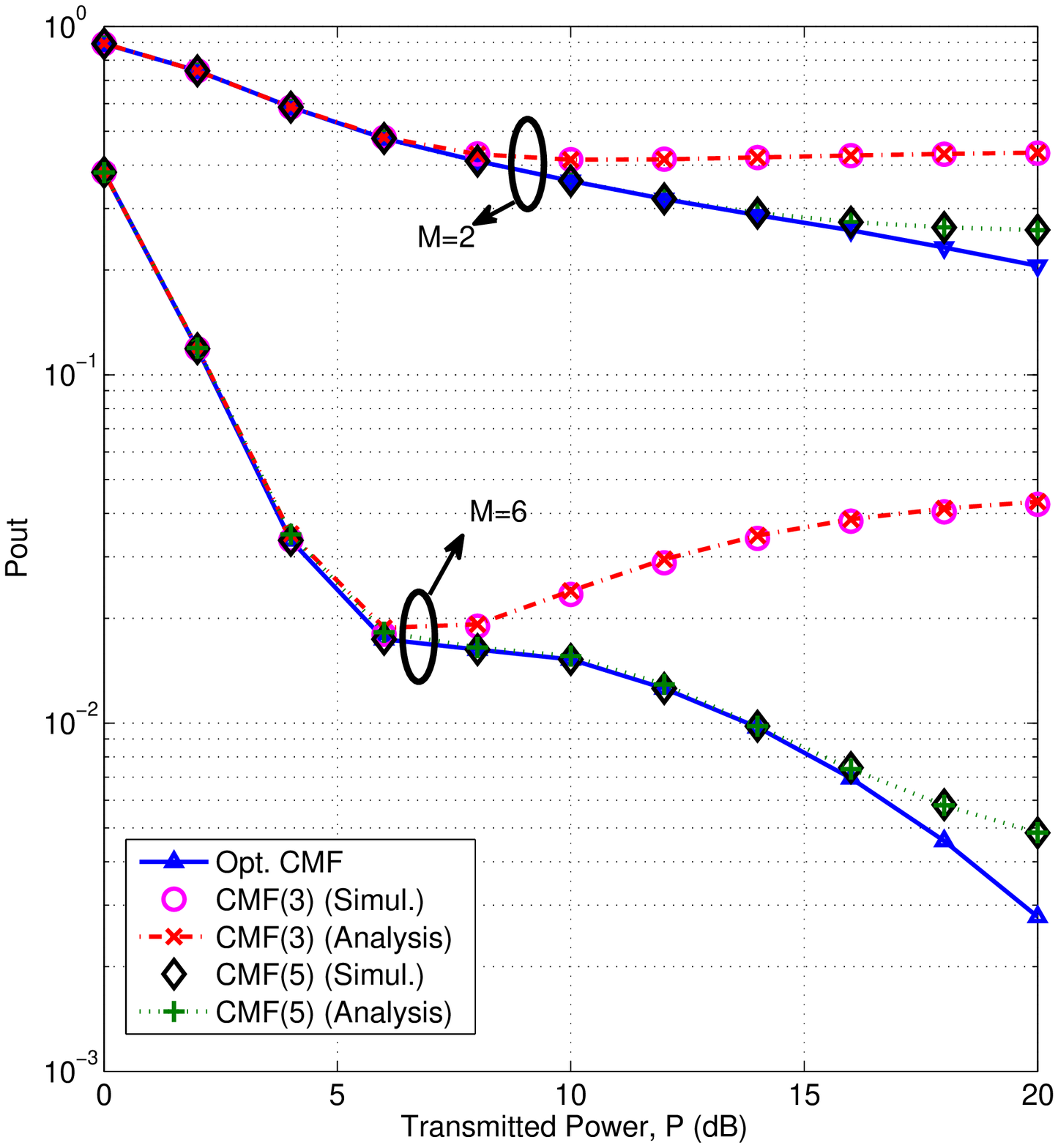}
\caption{System outage probability for optimum CMF, CMF(3), and CMF(5) versus average SNR, ($M=2, 6$).}
\label{fig:PoutSys}
\end{figure}

Fig.~\ref{fig:Pout_Pdet} demonstrates the outage probability along with the probability of rank failure versus the average SNR, for the cases of optimum CMF and CMF(5). Here, $M=6$ relays are assumed. As mentioned in Subsection~\ref{sec:analysis}-C and can be observed from this figure, probability of rank failure lower bounds the outage probability in the both cases. The lower bound is closely tight in SNRs more than 12 dB. Therefore, rank failure is a major bottleneck in the system and degrades the performance considerably, even in the optimum CMF method. From these results, by using a higher number of relays ($M>L$), the rank failure probability reduces which leads to a lower outage probability.\par
\begin{figure}[tb]
\renewcommand{\figurename}{Fig.}
\centering
\includegraphics[width = \columnwidth]{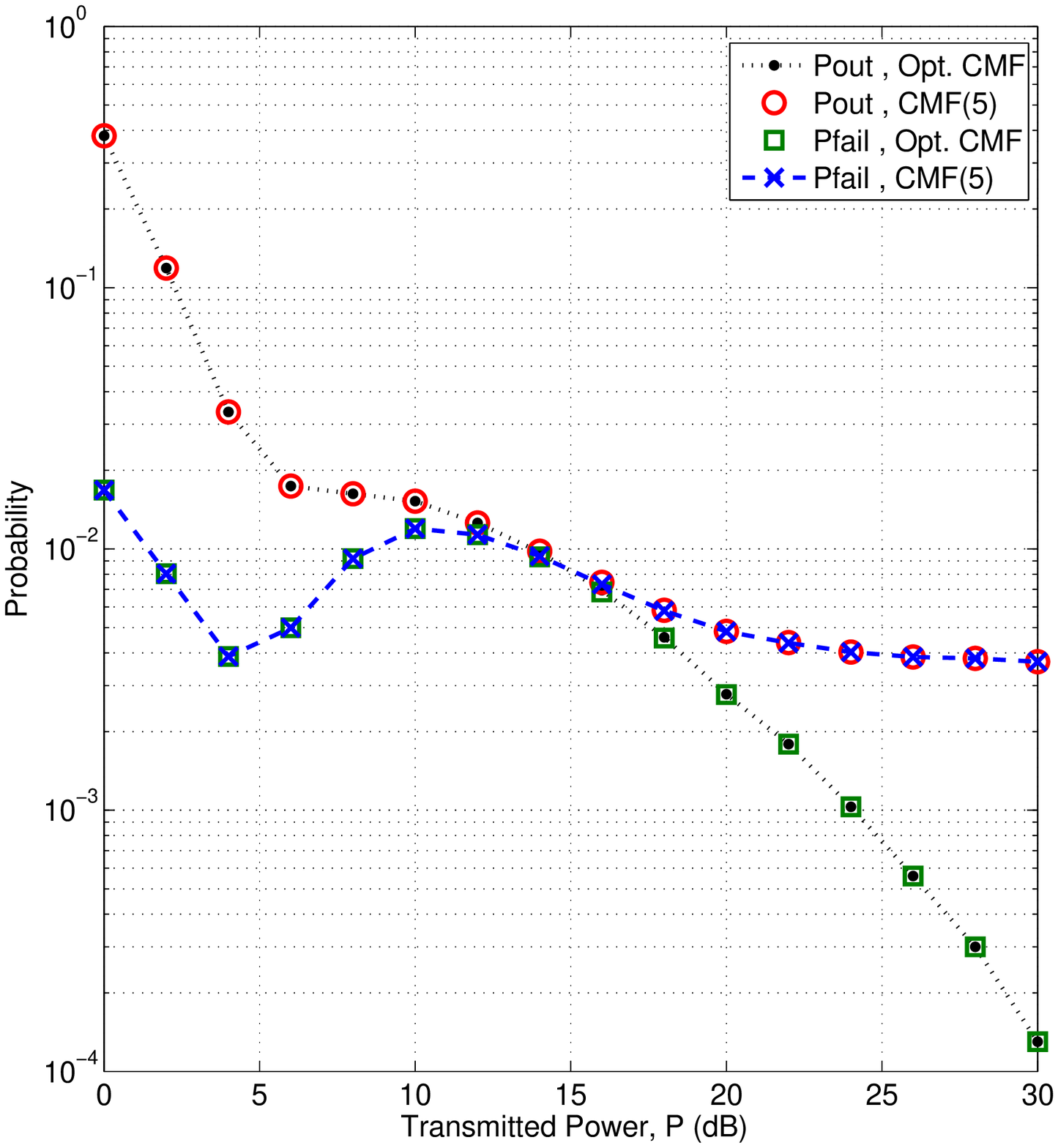}
\caption{Probabilities of outage and rank failure  for optimum CMF and CMF(5) versus average SNR, ($M=6$).}
\label{fig:Pout_Pdet}
\end{figure}

Fig.~\ref{fig:CEE} illustrates the effect of CEE on the optimum and simplified CMF methods. In this figure, the outage probabilities of optimum CMF and CMF(5) versus the average SNR are shown for different values of CEE variance $\sigma_e^2$ through computer simulations. The number of relays ($M$) is assumed to be 6. In the simulations, the true channel gain $h_{ml}$ and the estimated channel gain ${{\hat h}_{ml}}$ are computed as
\begin{align}
  {h_{ml}} &= \left| {{\gamma _{ml}}} \right|,  \\
  {{\hat h}_{ml}} &= \left| {{\gamma _{ml}} + {\sigma _e}{e_{ml}}} \right| ,
\end{align} 
where $\gamma _{ml}$ and $e_{ml}$, for $l=1,2,\, m=1,\cdots,M,$ are \textit{i.i.d} circularly-symmetric complex normal random variables with unit variance. CEE results in a noisy matrix ${\mathbf{G}}_m$ in~(\ref{eq:opt2}) and~(\ref{eq:Opt_sim}) and hence, a suboptimum ECV may be selected as the solution of these optimizations. Selection of a suboptimum ECV causes a rate loss. It can be found from this figure that the optimum CMF is more sensitive to CEE than the simplified CMF, and even in some cases shows slightly inferior performance than the simplified CMF. The reason is that optimum CMF has a larger search space, especially in high SNR values, than the simplified CMF, which has a small and fixed search space. Hence, when ${\mathbf{G}}_m$ is noisy due to CEE, it is more probable that a suboptimum ECV is selected as the solution of~(\ref{eq:opt2}) than~(\ref{eq:Opt_sim}).\par
\begin{figure}[tb!]
\renewcommand{\figurename}{Fig.}
\centering
\includegraphics[width = \columnwidth]{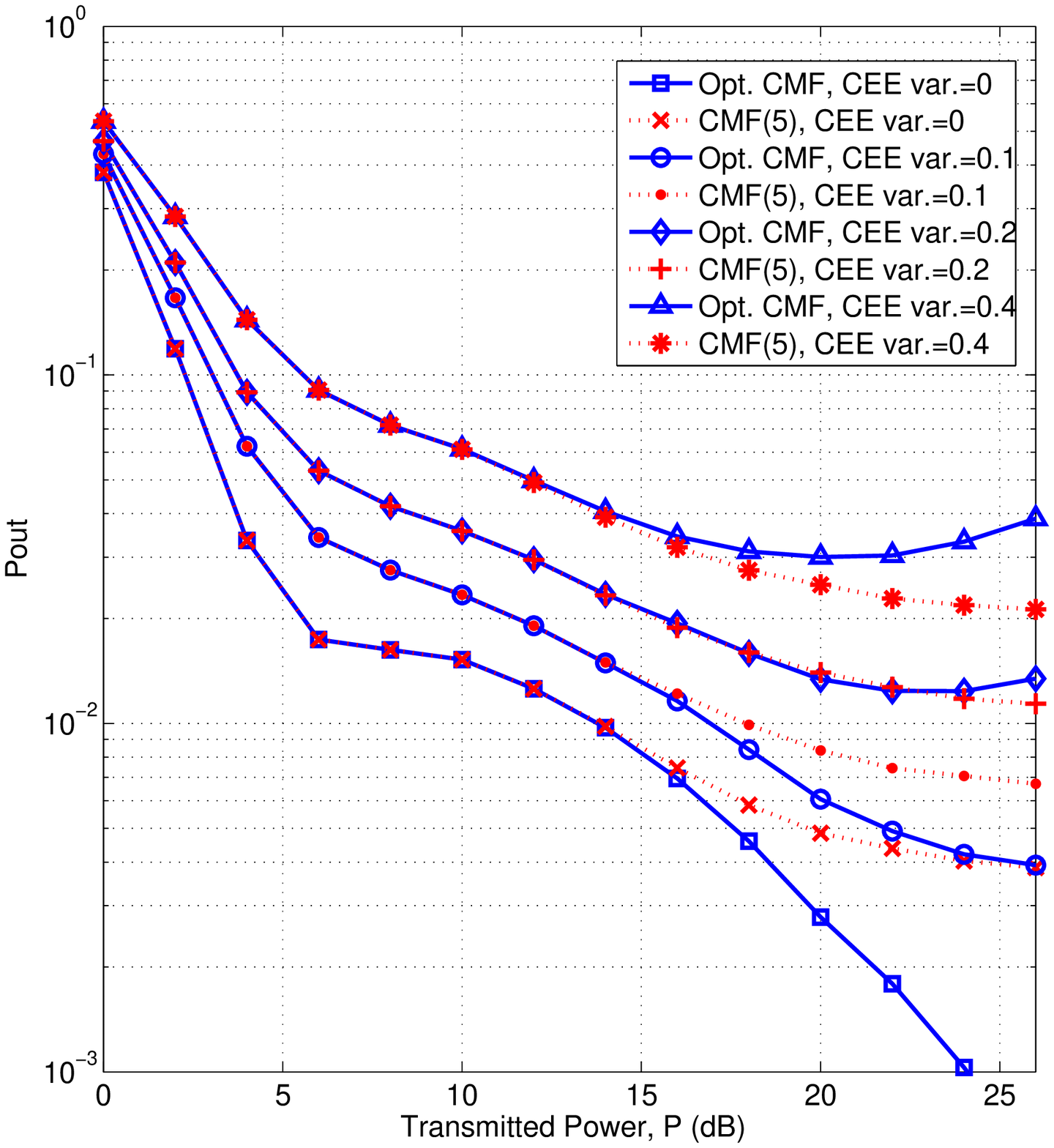}
\caption{Effect of CEE on outage probability for optimum CMF and CMF(5) versus average SNR, ($M=6$).}
\label{fig:CEE}
\end{figure}

\section{CONCLUSION} \label{sec:conclusion}
In this paper, we have proposed a simplified CMF method, called CMF($K$). Through analytical and simulation evaluations, it has been demonstrated that the CMF($K$) presents nearly the same performance as optimum CMF in low SNR regimes, where the SNR is below a certain threshold. Higher values of $K$ increase this threshold and extend the valid approximation region. The exact outage probability of CMF($K$) is derived. Here, we have considered two sources and real channels. However, the same framework can be applied to extend the results to a greater number of sources or complex cases. Our results indicates that the rank failure of the received equations at the destination is a dominate performance degradation in term of the overall system outage probability in both optimum and our proposed simplified CMF, which its effect can be reduced by having higher number of relays compared to the number of sources. Finally, computer simulations showed that simplified CMF, due to its small and fixed search space, is more robust against channel estimation error than optimum CMF for the examples considered.\par

\bibliographystyle{IEEEtr}
\bibliography{Paper2_Refs}

\end{document}